\documentclass[a4paper]{article}

\usepackage{fullpage}
\usepackage{amssymb,amsmath}
\usepackage{graphicx}
\usepackage{url}
\usepackage{amsthm}

\newtheorem{theorem}{Theorem}
\newtheorem{lemma}[theorem]{Lemma}

\urldef{\mailsa}\path|{otfried,|
\urldef{\mailsb}\path|pittzzang}@kaist.ac.kr|
\newcommand{\keywords}[1]{\par\addvspace\baselineskip
\noindent Keywords:\enspace\ignorespaces#1}

\usepackage{hyperref}

\newcommand{\Ts}{T^{\ast}}

\newcommand{\eps}{\varepsilon}

\title{Single-Source Dilation-Bounded\\ Minimum Spanning Trees%
  \thanks{This research was supported in part by NRF
    grant~2011-0016434 and in part by NRF grant 2011-0030044
    (SRC-GAIA), both funded by the government of Korea.}}

\let\geq\geqslant
\let\leq\leqslant

\makeatletter
\partopsep\z@ \textfloatsep 10pt plus 1pt minus 4pt
\def\section{\@startsection {section}{1}{\z@}{-3.5ex plus -1ex minus
-.2ex}{2.3ex plus .2ex}{\large\bf}}
\def\subsection{\@startsection{subsection}{2}{\z@}{-3.25ex plus -1ex
minus -.2ex}{1.5ex plus .2ex}{\normalsize\bf}}
\def\@fnsymbol#1{\ensuremath{\ifcase#1\or *\or **\or 1\or 2\or
    3\or 4\or 5\or 6\or 7 \or 8\ or 9 \or 10\or 11 \else\@ctrerr\fi}}
\makeatother

\author{Otfried Cheong%
  \thanks{Department of Computer Science, KAIST, Daejeon, South Korea.
    \mailsa \mailsb}
  \and Changryeol Lee\footnotemark[2]}

\begin{document}

\maketitle

\begin{abstract}
  Given a set $S$ of points in the plane, a geometric network for $S$
  is a graph $G$ with vertex set $S$ and straight edges.  We consider
  a broadcasting situation, where one point $r \in S$ is a designated
  source.  Given a dilation factor~$\delta$, we ask for a geometric
  network~$G$ such that for every point $v \in S$ there is a path from
  $r$ to $v$ in~$G$ of length at most $\delta|rv|$, and such that the
  total edge length is minimized.  We show that finding such a network
  of minimum total edge length is NP-hard, and give an approximation
  algorithm.
\end{abstract}

\section{Introduction}

Given a set~$S$ of points in the plane, a geometric network for~$S$ is
an edge-weighted graph~$G$ with vertex set~$S$ and straight edges.
The weight of an edge~$(u,v)$ is the length of the segment $uv$, that
is, the Euclidean distance $|uv|$ of the two points.

Various types of networks, such as communication networks, road
networks, or telephone networks, have been modeled as geometric
networks.  One important parameter of a geometric network is its total
\emph{cost} $\ell(G)$: the sum of all edge lengths.   The network that minimizes
the cost while connecting all points in~$S$ is the Euclidean minimum
spanning tree of~$S$.  Another well-studied parameter is the
\emph{dilation} of a network.  For two points $u,v \in S$, the
dilation of the pair~$(u,v)$ is defined to be the ratio
\begin{displaymath}
    \Delta_G(u,v):= {d_G(u,v) \over |uv|}.
\end{displaymath}
of the length of the shortest path between $u$ and~$v$ in~$G$ and the
distance~$|uv|$.  The dilation of a network is commonly defined as the
maximum of $\Delta_G(u, v)$ over all pairs $u, v \in S$.  The network
minimizing the dilation is the complete graph, which has dilation~1.

The complete graph has prohibitively large cost, while the minimum
spanning tree may have large dilation.  Balancing these two parameters
has been the subject of much research in the literature, we refer to
the book by Narasimhan and Smid~\cite{narasimhan-smid} for an
overview.

In this paper, we consider a \emph{broadcasting} situation, where one
point $r \in S$ is a designated source, and the purpose of the network
is to broadcast information from~$r$ to all the other nodes.  A node
$v$ receives the information with delay~$d_G(r,v)$, and we are
interested in the \emph{relative delay}, which is the dilation
$\Delta_G(r,v) = d_G(r,v)/|rv|$.  We define the (relative)
\emph{delay} of $G$ to be the maximum
\[
\Delta(G) = \max_{v \in S\setminus \{r\}} \Delta_G(r, v)
= \max_{v \in S\setminus \{r\}} \frac{d_G(r, v)}{|rv|}.
\]
Note that we can assume our network~$G$ to be a tree, as the
shortest-path tree with source at~$r$ will have the same delay as~$G$
itself.

The network minimizing $\Delta(G)$ is the star consisting of all edges
$(r, v)$, for $v\in S\setminus\{r\}$.  It has $\Delta(G) = 1$, but its
cost $\ell(G)$ may be prohibitively large.  The minimum spanning tree,
which minimizes~$\ell(G)$, may have large delay.  In practice, a
trade-off between both factors needs to be achieved. Related notions
have been studied in the literature. For instance, computing the
tree~$T$ that minimizes $\ell(T)$ while bounding the distance of all
nodes from the source in the graph is
NP-complete~\cite{pyo1996constructing}.  Similarly,
minimizing~$\ell(T)$ while keeping all nodes at most~$k$ hops from
the source is NP-complete~\cite{althaus2005approximating}.

Our problem, which we term the \emph{single-source dilation-bounded
  minimum spanning tree} problem, takes as input the point set~$S$,
the designated source point~$r \in S$, and a delay~bound~$\delta \geq
1$.  We ask for a spanning tree~$T$ of smallest possible
cost~$\ell(T)$ such that $\Delta(T) \leq \delta$, or, in other words,
such that for every $v \in S \setminus \{r\}$, we have $d_T(r,v) \leq
\delta|rv|$.

When we set $\delta = 1$, then the answer is simply the star
connecting $r$ with all other points.  When we set $\delta \geq n-1$,
then the answer is the minimum spanning tree of~$S$, and so our
problem interpolates between these two networks.

We show that solving the problem exactly is NP-hard by reduction from
Knapsack in Section~\ref{sec:np-hard}.  This leads us to consider an
approximation algorithm.  Since this algorithm is rather simple, we
present it first in Section~\ref{sec:approx}.

\section{Approximation Algorithm}
\label{sec:approx}

An approximate solution to our problem can be computed using known
algorithms for $\delta$-spanners.  We introduce these results first.

Given a point set $S$ in the plane, a $\delta$-spanner for $S$ is a
geometric network on~$S$ such that $d_{G}(u,v) \leq \delta|uv|$ for
each pair of points $u,v \in S$, where $\delta \geq 1$.  Many
algorithms to compute a $\delta$-spanner for a given point set have
been given in the literature, see again the book by Narasimhan and
Smid~\cite{narasimhan-smid} for an overview.  We will make use of a
result by Gudmundsson, Levcopoulos, and
Narasimhan~\cite{gudmundsson2000improved}, who showed the following:

For any real number $\delta>1$, a $\delta$-spanner of a point set~$S$
in $\mathbb{R}^d$ can be constructed in $O(n\log n)$ time such that
the spanner has $c_1 n$ edges, maximum degree~$c_2$, and cost $c_3
\ell(M)$, where $M$ is the minimum spanning tree of~$S$, and $c_1,
c_2, c_3$ are constants depending on~$\delta$ and the dimension~$d$.

We can now explain our simple approximation algorithm:
\begin{theorem}
  Given a set of $n$ points $S$ in the plane, a designated source~$r
  \in S$, and a real constant $\delta>1$, we can construct in time
  $O(n \log n)$ a tree~$T$ with vertex set~$S$ such that the delay
  $\Delta(T) \leq \delta$, and $\ell(T) \leq c_\delta\ell(M)$, where
  $M$ is the minimum spanning tree of~$S$ and $c_\delta$ is a constant
  depending on~$\delta$.
\end{theorem}
\begin{proof}
  We first run the algorithm by Gudmundsson et
  al.~\cite{gudmundsson2000improved} to obtain a $\delta$-spanner~$G$
  for~$S$, with total cost $\ell(G) \leq c_\delta \ell(M)$, where
  $c_\delta$ is a constant that depends on~$\delta$.

  We then compute the shortest path tree~$T$ with source~$r$ in~$G$.
  Clearly we have $\ell(T) \leq \ell(G) \leq c_\delta\ell(M)$.
  For any $v \in S \setminus \{r\}$, we have $d_T(r,v) = d_G(r,v) \leq
  \delta |rv|$ since $G$ is a $\delta$-spanner, and so the delay
  $\Delta(T) \leq \delta$.

  $G$ can be computed in time $O(n\log
  n)$~\cite{gudmundsson2000improved}, and the shortest-path tree~$T$
  can be computed in time $O(n \log n)$, for instance using Dijkstra's
  algorithm, using the fact that $G$ has $O(n)$ edges.
\end{proof}

\section{Single-Source Dilation-Bounded Minimum Spanning Tree is
  NP-hard}
\label{sec:np-hard}

In this section, we show that the decision version of our problem is NP-hard.
\begin{theorem}
  \label{thm:np-hard}
  Given a set $S$ of points in the plane, a designated source~$r \in
  S$, a delay bound $\delta \geq 1$, and a cost bound $K \geq 0$,
  it is NP-hard to decide if there exists a tree~$T$ for $S$ with
  delay $\Delta(T) \leq \delta$ and cost $\ell(T) \leq K$.
\end{theorem}
The proof is by reduction from \textsc{Knapsack}, which we define
first.  \textsc{Knapsack} is well known to be
NP-complete~\cite{garey1979computers}.
\begin{quotation}
  \textsc{Knapsack}: Given a set of $n$~items, each having an integer
  profit~$p_i > 0$ and an integer weight~$w_i > 0$, as well as a
  profit bound~$P > 0$ and a weight bound~$W > 0$.  Does there exist a
  subset of items with total weight at most~$W$ and total profit at
  least~$P$?
\end{quotation}
Our reduction takes as input an instance of \textsc{Knapsack}, and
produces an instance of single-source dilation-bounded minimum
spanning tree.  For a set of~$n$ items with profits $p_{i}$ and
weight~$w_{i}$, we construct a set~$S$ of $3n + 4$ points.  For each
item~$i$, we construct three points~$a_{i}, b_{i}, c_{i}$.  In
addition, we create three extra points $d_0$, $d_1$, $d_2$, as well as
a source~$r$.

We first define
\begin{align*}
  \alpha_i & := p_i+w_i > 0\\
  \beta_i & := 2p_i+w_i=\alpha_i+p_i > \alpha_i\\
  \gamma_i & := 3p_i+w_i=\beta_i+p_i > \beta_i.
\end{align*}
We thus have
\[
\alpha_i+\beta_i=3p_i+2w_i=\gamma_i+w_i>\gamma_i,
\]
and so $0<\alpha_i<\beta_i<\gamma_i<\alpha_i+\beta_i$.  This implies
that the three values $\alpha_i$, $\beta_i$ and~$\gamma_i$ satisfy the
triangle inequality.
We also set
\[
m := \max_{1 \leq i \leq n} \gamma_i,\qquad L:=\sum_{i=1}^{n}(\gamma_i+m),
\]

\paragraph{Construction of $S$.}
Our construction starts by placing the source~$r$ at the origin.  We
then place the $2n$~points~$a_{i}$ and~$b_{i}$ on the negative
$y$-axis as follows (see Fig.~\ref{fig:construction}):
\begin{align*}
  a_i &:= \Big(0, -4L-\textstyle\sum_{j=1}^{i-1}(\gamma_{j}+m)\Big)\\
  b_i &:= a_i - (0, \gamma_i)
\end{align*}
\begin{figure}[htb]
  \centerline{\includegraphics{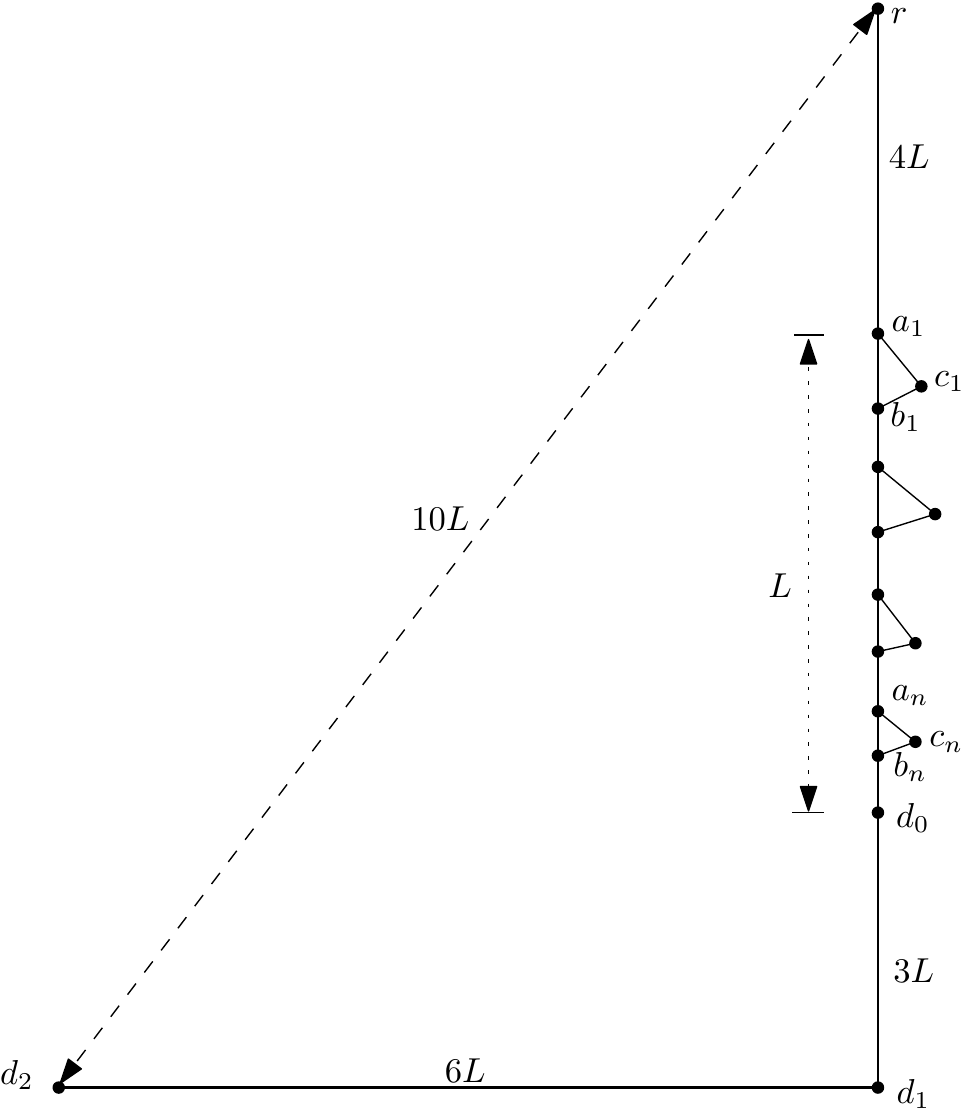}}
  \caption{The construction of $S$}
  \label{fig:construction}
\end{figure}

We have $|a_ib_i| = \gamma_i$ and $|b_ia_{i+1}| = m$.  We can now
place $c_{i}$, for $1 \leq i \leq n$ on the right side of the $y$-axis
such that $|a_ic_i| = \beta_i$ and $|c_ib_i| = \alpha_i$.
Finally, we define $d_{0}$, $d_{1}$, and $d_{2}$ as follows:
\begin{align*}
  d_0 &:= (0, -5L) \\
  d_1 &:= (0, -8L) \\
  d_2 &:= (-6L, -8L)
\end{align*}

\paragraph{Regular trees.}

We classify the edges connecting points of~$S$ into \emph{regular} and
\emph{irregular} edges.  Regular edges are the edges $ra_1$,
$b_{n}d_{0}$, $d_0d_1$, $d_1d_2$, the edges $a_ib_i$, $b_ic_i$,
$a_ic_i$, for $1 \leq i \leq n$, and the edges $b_{i}a_{i+1}$ for $1
\leq i < n$.  A \emph{regular tree} is a spanning tree on~$S$ which
contains only regular edges.  A tree containing irregular edges is
irregular.

In a regular tree~$T$, the dilation of the pair $(r, d_2)$ is much
larger than the dilation of any other pair~$(r, v)$, and so we have
\begin{lemma}
  \label{lem:regular}
  The delay of a regular tree~$T$ is $\Delta(T) = \Delta_T(r, d_2)$.
\end{lemma}
\begin{proof}
Let $L_i=|a_1a_i|=\sum_{j=1}^{i-1}(\gamma_i+m)$.  Since $T$ is a tree,
there is a unique path from $b_i$ to $a_{i+1}$ and we have $d_T(a_i,
b_i) \leq |a_ic_i|+|c_ib_i|$.
Since $\alpha_i<\beta_i<\gamma_i \leq m$, we have
\begin{align*}
  d_T(r,b_i) &\leq
  |r,a_1|+\sum_{j=1}^i(|a_jc_j|+|c_jb_j|)+\sum_{j=1}^{i-1}|b_ja_{j+1}|\\
  & = 4L+\sum_{j=1}^i(\beta_j+\alpha_j)+\sum_{j=1}^{i-1}m < 4L+L+L_i=5L+L_i\\
  d_T(r,c_i)&\leq
  |r,a_1|+\sum_{j=1}^{i-1}(|a_jc_j|+|c_jb_j|+|b_ja_{j+1}|)+|a_ib_i|+|b_ic_i|\\
  &= 4L+\sum_{j=1}^{i-1}(\beta_j+m)+(\sum_{j=1}^i\alpha_j+\gamma_i)<4L+L_i+L=5L+L_i
\end{align*}
Note that $d_T(r,b_i)=d_T(r,a_i)+d_T(a_i,b_i)>d_T(r,a_i)$. Thus, we get
\begin{align*}
  \Delta_T(r, a_i) &= {{d_T(r,a_i)} \over |ra_i|}
  < {{d_T(r,b_i)} \over |ra_i|}
  \leq {{5L+L_i} \over {4L+L_i}} \leq 1.25\\
  \Delta_T(r, b_i)&= {{d_T(r,b_i)} \over |rb_i|}
  < {{d_T(r,b_i)} \over |ra_i|} \leq 1.25.
\end{align*}
Since $|rc_i|>|ra_i|$, we also have
\begin{align*}
\Delta_T(r, c_i) &= {{d_T(r,c_i)} \over |rc_i|}
< {d_T(r,c_i) \over |ra_i|}={{5L+L_i} \over {4L+L_i}}\leq 1.25 \\
\Delta_T(r, d_0) &= {d_T(r,d_0) \over |rd_0|}={d_T(r,b_n)+m \over
  5L}<{5L+L_n+m \over 5L}< {6L \over 5L}< 1.25\\
\Delta_T(r, d_1) &= {{d_T(r,d_1)} \over {|rd_1|}}={{d_T(r,d_0)+|d_0d_1|}
  \over {|rd_0|+|d_0d_1|}}
\leq{{6L+|d_0d_1|} \over 5L+|d_0d_1|}< 1.25.
\end{align*}
On the other hand,
\[
\Delta_T(r, d_2) = {d_T(r,d_2) \over |rd_2|}\geq {|rd_1|+|d_1d_2| \over
  |rd_2|}= {14L \over 10L}=1.4
\]
Hence, the delay of~$T$ is determined by~$d_2$:
\[
\Delta(T)=\max_{v \in S \setminus\{r\}} \Delta_{T}(r, v) =
\Delta_T(r, d_2) \qedhere
\]
\end{proof}

We define a special regular tree, the \emph{base tree}~$T_{0}$, which
contains all regular edges except for the edges~$a_{i}c_{i}$.
Delay and cost of the base tree are as follows.
\begin{lemma}
  \label{lem:base}
  The total edge length of the base tree is $\ell(T_{0}) < 14.5 L$,
  and its delay is $\Delta(T_{0}) = 1.4$.
\end{lemma}
\begin{proof}
  The total edge length and the delay of the base tree $T_0$ are
  \begin{align*}
    \ell(T_0) &= |rd_1|+\sum_{i=1}^{n}(|b_ic_i|)+|d_1d_2|\\
    &= 8L+{1 \over 2}\sum_{i=1}^{n}2\alpha_i+6L
    < 14L+{1 \over 2}\sum_{i=1}^{n}(\gamma_i+m)=14.5L \\
    \Delta(T_0) & = \Delta_{T_0}(r, d_2)= {{|rd_1|+|d_1d_2|} \over {|rd_2|}}
    = {{8L+6L} \over {10L}}
    = 1.4 \qedhere
  \end{align*}
\end{proof}
Connecting~$d_{2}$ to any point other than~$d_{1}$ will always produce
trees with higher cost than the base tree:
\begin{lemma}
  \label{lem:exclude-d2}
  If $T$ is an irregular tree that contains an edge $vd_2$ for $v
  \neq d_1$, then $\ell(T) > \ell(T_{0})$.
\end{lemma}
\begin{proof}
  Assume $T$ contains an edge $vd_{2}$, for $v \neq d_1$, and consider
  the path~$Q$ in~$T$ connecting~$r$ and~$d_{1}$.  There are two
  possible cases.  First, assume that $Q$ passes through~$d_2$. Then
  we have
  \[
  \ell(T) \geq \ell(Q) \geq |rd_2|+|d_2d_1|=10L+6L=16L>\ell(T_0).
  \]
  In the second case, $Q$ does not pass through~$d_2$.  Since $|vd_2|
  \geq |d_0d_2|$ and $\ell(Q) \geq |rd_1|$, we have
  \[
  \ell(T) \geq \ell(Q)+|vd_2| \geq |rd_1|+|d_0d_2|=8L+3\sqrt{5}L \simeq
  14.7082L > \ell(T_0). \qedhere
  \]
\end{proof}

We can now show that for $\delta \geq 1.4$, regular trees are
better than irregular trees.
\begin{lemma}
  \label{lem:regular-best}
  For every irregular tree~$T_1$ with $\Delta(T_1) \geq 1.4$, there
  exists a regular tree $\Ts$ such that $\ell(T_1) \geq \ell(\Ts)$ and
  $\Delta(T_1) \geq \Delta(\Ts)$.
\end{lemma}
\begin{proof}
  If $T_1$ includes an edge $vd_{2}$ with $v \neq d_1$, then
  Lemma~\ref{lem:exclude-d2} implies the lemma with $\Ts = T_0$.
  We can therefore assume that $T_1$ contains the edge $d_1d_2$, and no
  other edge incident to~$d_{2}$.

  For a spanning tree~$T$ of~$S$, let $\pi(T)$ denote the path
  connecting~$r$ and~$d_{1}$ in~$T$.  We define $J(T)$ to be the set
  of indices $i \in \{1,\dots,n\}$ such that $\pi(T)$ contains all
  three vertices~$a_{i}$, $b_{i}$, and~$c_{i}$, but does not contain
  both edge~$a_{i}c_{i}$ and~$b_{i}c_{i}$.

  Let $\mathcal{T}$ denote the set of spanning trees~$T$ of~$S$ such
  that $d_2$ is adjacent only to $d_{1}$ in~$T$, $\ell(T) \leq
  \ell(T_1)$, and $\ell(\pi(T)) \leq \ell(\pi(T_1))$.  We have $T_1
  \in \mathcal{T}$, so $\mathcal{T} \neq \emptyset$.

  We now pick a spanning tree~$T_2\in \mathcal{T}$ such that $J(T_2)$
  is minimal under inclusion.  Let us suppose first that $J(T_2) \neq
  \emptyset$.  Then there is an $i \in J(T_2)$ such that $\pi(T_2)$
  contains $\{a_i,b_i,c_i\}$, but does not contain both $a_{i}c_{i}$
  and $b_{i}c_{i}$.  Let $u$ be the first of the three vertices
  encountered by~$\pi(T_2)$, let $v$ be the second one, and let $w$ be
  the last one.  One of the two edges incident to~$v$ in~$\pi(T_2)$ is
  different from $a_{i}c_{i}$ and $b_{i}c_{i}$.  Denote this edge
  by~$vp$.  Then $|vp| \geq \gamma_{i} \geq |uw|$.  We obtain a new
  spanning tree $T_3$ from~$T_2$ by adding~$uw$ and removing~$vp$.  We
  have $\ell(T_3) = \ell(T_2) + |uw| - |vp| \leq \ell(T_2) \leq
  \ell(T_1)$.  Clearly $\ell(\pi(T_3)) \leq \ell(\pi(T_2))$, and so
  $T_3 \in \mathcal{T}$.  Furthermore, $J(T_3) \subsetneq J(T_2)$ as
  $i \not\in J(T_3)$, a contradiction to the choice of~$T_2$.

  It follows that $J(T_2) = \emptyset$.  Let us define the set $I
  \subset \{1,\dots,n\}$ of indices~$i$ such that $\pi(T_2)$ contains
  both edge~$a_ic_i$ and~$c_ib_i$.  We define $\Ts$ as the tree
  consisting of all regular edges, except that we remove $a_{i}b_{i}$
  when $i \in I$, and that we remove $a_{i}c_{i}$ when $i \not\in
  I$. Fig.~\ref{fig:regular-best} shows an example of the regular tree
  we construct.
  \begin{figure}[h]
    \centerline{\includegraphics{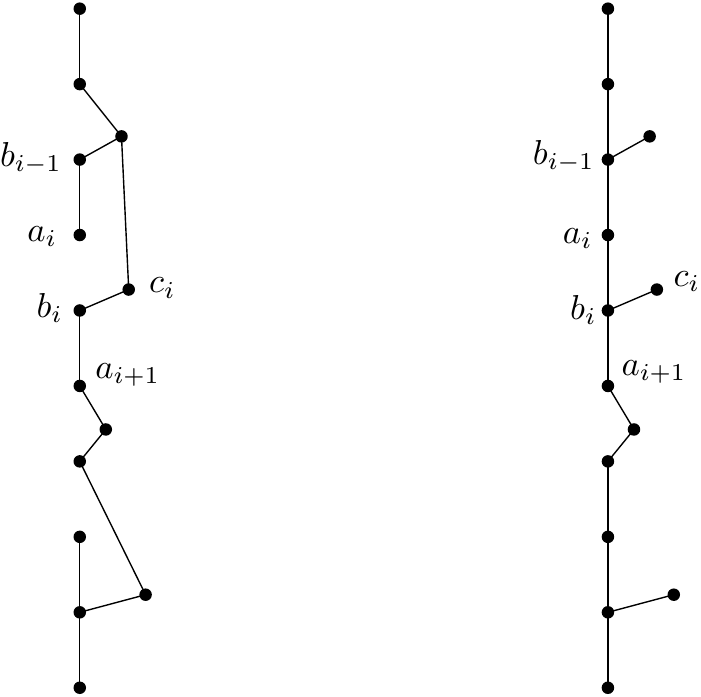}}
    \caption{Irregular tree $T_2$ (left) and its regular tree $\Ts$ (right)}
    \label{fig:regular-best}
  \end{figure}
  Let $E_{0}$ denote the set of edges $a_{i}c_{i}$ and $c_{i}b_{i}$
  for all $i \in I$.  By definition of~$I$, $E_{0} \subset \pi(T_2)$,
  and we define $E_{1} = \pi(T_2) \setminus E_{0}$.  Let
  $E_{0}^{\ast}$ and $E_{1}^{\ast}$ be the projection of $E_{0}$ and
  $E_{1}$ on the~$y$-axis. Then $E_{0}^{\ast} \cup E_{1}^{\ast}$ must
  be equal to the segment~$rd_1$, and so $\ell(E_{1}) \geq
  \ell(E_{1}^{\ast}) \geq |rd_1| - \ell(E_0^{\ast})$.  It follows that
  \[
  \ell(\pi(\Ts)) = |rd_1| - \ell(E_0^{\ast}) + \ell(E_0)
  \leq \ell(E_1) + \ell(E_0) = \ell(\pi(T_2)) \leq \ell(\pi(T_1)).
  \]
  This implies that $\Delta(T_1) \geq \Delta_{T_1}(r, d_2) \geq
  \Delta_{\Ts}(r, d_2) = \Delta(\Ts)$ by Lemma~\ref{lem:regular}.

  For each point $v \in S \setminus \{r\}$, let $p(v)$ be the second
  vertex on the path from~$v$ to~$r$ in~$T_2$.  We have $\ell(T_2) =
  \sum_{v \in S\setminus\{r\}} |vp(v)|$.  For the vertices
  on~$\pi(T_2)$, we have $\sum_{v \in \pi(T_2)\setminus\{r\}}|vp(v)| =
  \ell(\pi(T_2))$.  Since $J(T_2)=\emptyset$, for each $i \not\in I$,
  one of the three vertices $a_{i}$, $b_{i}$, $c_{i}$ is \emph{not} on
  the path~$\pi(T_2)$.  Since its nearest vertex has distance at
  least~$\alpha_i$, we have
  \[
  \ell(T_2) = \sum_{v \in S\setminus\{r\}} |vp(v)|
  \geq \ell(\pi(T_2)) + \sum_{i \not\in I}\alpha_i + |d_2d_1|.
  \]
  On the other hand, we have
  \begin{align*}
  \ell(\Ts) & = \ell(\pi(\Ts)) + \sum_{i\not\in I}\alpha_i + |d_2d_1| \\
  & \leq \ell(\pi(T_2)) + \sum_{i\not\in I}\alpha_i + |d_2d_1|
  \leq \ell(T_2) \leq \ell(T_1).  \qedhere
  \end{align*}
\end{proof}

\paragraph{Correctness of the reduction.}

It remains to show that the constructed point set~$S$ has a spanning
tree of small delay and small cost if and only if the original
\textsc{Knapsack} instance had a positive answer.
\begin{lemma}
  \label{lem:reduction}
  The \textsc{Knapsack} instance has a positive answer if and only if
  there is a spanning tree~$T$ for~$S$ with delay $\Delta(T) \leq 1.4
  + (W / 10L)$ and cost $\ell(T) \leq \ell(T_{0}) - P$.
\end{lemma}
\begin{proof}
  We first assume that the \textsc{Knapsack} instance has a positive
  answer.  Let $I \subseteq \{1,2,\ldots,n\}$ be a set of indices such
  that $\sum_{i \in I}p_{i} \geq P$ and $\sum_{i \in I}w_{i} \leq W$.
  Let $T$ be the tree consisting of all regular edges, except that we
  exclude $a_{i}b_{i}$ for $i \in I$, and exclude $a_{i}c_{i}$
  for $i \not\in I$.

  Then we have
  \begin{align*}
    \ell(T) & = \ell(T_{0}) - \sum_{i \in I}(\gamma_i -
    \beta_i) = \ell(T_{0}) - \sum_{i \in I}p_{i} \leq \ell(T_{0}) - P, \\
    d_T(r, d_2) & = |rd_1| + |d_1d_2| + \sum_{i \in
      I}(\alpha_i + \beta_i - \gamma_i)
    = 14L + \sum_{i \in I} w_{i} \leq 14L + W \\
    \Delta(T) & = \Delta_T(r, d_2) \leq \frac{14L + W}{10L} = 1.4 + (W/10L),
  \end{align*}
  and the claim follows.

  Assume now that $T$ is a spanning tree for~$S$ with the given
  bounds.  If $\Delta(T) < 1.4$, then $T$ must include an edge
  incident to $d_2$ other than $d_1d_2$ and is not regular. But then
  Lemma~\ref{lem:exclude-d2} implies that $\ell(T)> \ell(T_0)$, a
  contradiction.  So $\Delta(T) \geq 1.4$, and by
  Lemma~\ref{lem:regular-best} we can assume that $T$ is regular.

  Since $T$ is a spanning tree, it must include all regular edges,
  except that for each~$1 \leq i \leq n$, one of the three edges
  $a_{i}b_{i}$, $a_{i}c_{i}$, or $b_{i}c_{i}$ must be missing.  We
  define $I \subset \{1,2,\dots,n\}$ to be the set of indices~$i$ such
  that $T$ does \emph{not} include the edge~$a_{i}b_{i}$.

  We have
  \[
  d_T(r, d_2) = |rd_1| + |d_1d_2| + \sum_{i\in I}(\alpha_i + \beta_i -
  \gamma_{i}) = 14L + \sum_{i \in I} w_{i}.
  \]
  Since $\Delta(T)= \Delta_{T}(r,d_2) = d_T(r,d_2)/10L$, we have
  \[
  \sum_{i\in I}w_{i} = d_T(r,d_2) - 14L = 10L\Delta(T) - 14L
  \leq 10L(1.4 + (W/10L)) - 14L = W.
  \]

  The cost of $T$ is
  \[
  \ell(T) \geq \ell(T_0) - \sum_{i \in I}(\gamma_i - \beta_i)
  = \ell(T_0) - \sum_{i\in I}p_{i},
  \]
  and so
  \[
  \sum_{i\in I}p_{i} \geq \ell(T_{0}) - \ell(T)
  \geq \ell(T_0) - (\ell(T_0) - P) = P.
  \]
  It follows that the \textsc{Knapsack} instance has a positive
  answer.
\end{proof}

\paragraph{Reduction with integer coordinates.}

To complete our proof of Theorem~\ref{thm:np-hard}, we need to
construct a set of points with integer coordinates, such that the
total number of bits is polynomial in the size of \textsc{Knapsack}
instance. The construction given so far does not achieve this yet,
since the points $c_i$ are defined as the solution of a quadratic
equation. We will therefore compute \emph{approximations}
$\tilde{c_i}$ with $|c_i-\tilde{c_i}| < \eps$, for an $\eps$ to be
determined later. The set of points obtained in that way will be
denoted by $\tilde{S}$, which is the set of points $r$, $a_i$, $b_i$,
$\tilde{c_i}$, $d_0$, $d_{1}$, $d_{2}$. In the following lemma, we
bound by how much this approximation can change the tree cost and
delay.
\begin{lemma}
  \label{lem:error}
  If $T$ is a spanning tree on~$S$ and $\tilde{T}$ is the
  corresponding tree on~$\tilde{S}$, then
  $|\ell(T)-\ell(\tilde{T})|<12n\eps$, and
  $|\Delta(T)-\Delta(\tilde{T})|<20n\eps$.
\end{lemma}
\begin{proof}
  Let $u$, $v$ be a pair of points in $S$, with $\tilde{u}$,
  $\tilde{v}$ denoting the corresponding points in $\tilde{S}$.  Since
  $|u\tilde{u}| < \eps$ and $|v\tilde{v}| < \eps$, we have
  $||u\tilde{u}| -|v\tilde{v}|| < 2\eps$. The tree $T$ has $3n+3$
  edges, and so $|\ell(T)-\ell(\tilde{T})| < (6n+6)\eps \leq 12n \eps$.

  Consider now $X := d_T(r, v)$, $\tilde{X} := d_{\tilde{T}} (r,
  \tilde{v}$), $Y := |rv|$, and $\tilde{Y}:= |r\tilde{v}|$. Since
  $|v\tilde{v}| < \eps$, we have $|Y -\tilde{Y}| < \eps$. The path
  from~$r$ to $\tilde{v}$ in $\tilde{T}$ passes through most
  $n$~approximated points, and so $|X-\tilde{X}| < 2n\eps$. Since the
  longest edge in $T$ has length~$10L$ and the path has at most
  $3(n+1)$~edges, we have $X<30L(n+1)$. We also have that $Y \geq 4L$
  by the construction of~$S$. This means that $X/Y<7.5(n+1)$. We get
  \begin{align*}
    \frac{\tilde{X}}{\tilde{Y}} - \frac{X}{Y}
    & = \frac{\tilde{X}Y-X\tilde{Y}}{Y\tilde{Y}}
    < \frac{Y(X+2n\eps)-X(Y-\eps)}{Y\tilde{Y}}
    = {{2n\eps} \over \tilde{Y}}+ {\eps \over \tilde{Y}} \cdot
    {X \over Y}, \\
    {{X \over Y}-{\tilde{X} \over \tilde{Y}}}
    &= {{X\tilde{Y}-\tilde{X}Y} \over {Y\tilde{Y}}}
    < {{X(Y+\eps)-Y(X-2n\eps)} \over Y\tilde{Y}}
    = {{2n\eps} \over \tilde{Y}}+ {\eps \over \tilde{Y}} \cdot
    {X \over Y}.
  \end{align*}
  And since $\tilde{Y} \geq 1$,
  \[
  {{2n\eps} \over \tilde{Y}}+ {\eps \over \tilde{Y}} \cdot
  {X \over Y}
  \leq {2n\eps}+ {\eps} \times 7.5(n+1)<20n\eps. \qedhere
  \]
\end{proof}

If the \textsc{Knapsack} instance has a positive answer, then $S$ has
a spanning tree~$T$ with $\Delta(T) \leq 1.4 + (W/10L)$ and $\ell(T)
\leq \ell(T_0) - P$.  On the other hand, if the instance has a
negative answer, then this implies that for any subset of indices~$I
\subset \{1,2,\dots,n\}$ we have either $\sum_{i \in I}w_{i} \geq W +
1$ or $\sum_{i \in I} p_{i} \leq P - 1$.  By
Lemma~\ref{lem:reduction}, this means that any spanning tree~$T$
for~$S$ has either delay $\Delta(T) \geq 1.4 + ((W+1)/10L)$ or cost
$\ell(T) \geq \ell(T_{0}) - P + 1$.

Let us set $\eps = 1/600nL$.  We approximate the $c_i$ with a
precision of most~$\eps$, resulting in the point set~$\tilde{S}$.
This set is the input to our problem, with a delay bound of $\delta =
1.4 + (W/10L) + (1/20L)$, and a cost bound of $K = \ell(\tilde{T}_0) -
P + 0.5$.

If the \textsc{Knapsack} instance has a positive answer, then by
Lemma~\ref{lem:error}, there is a spanning tree~$\tilde{T}$
for~$\tilde{S}$ with $\Delta(\tilde{T}) \leq 1.4 + (W/10L) + 20n\eps =
1.4 + (W/10L) + (1/30L) < \delta$ and $\ell(\tilde{T}_0) -
\ell(\tilde{T}) > P - 24n\eps > P - 0.5$.

On the other hand, if the \textsc{Knapsack} instance has a negative
answer, then by Lemma~\ref{lem:error}, every spanning tree~$\tilde{T}$
for~$\tilde{S}$ has either $\Delta(\tilde{T}) \geq 1.4 + (W/10L) +
(1/10L) - 20n\eps = 1.4 + (W/10L) + (1/10L) - (1/30L) =
1.4 + (W/10L) + (1/15L) > \delta$, or we have
$\ell(\tilde{T}_0) - \ell(\tilde{T}) < P - 1 + 24n\eps < P - 0.5$.

In both cases, solving our single-source dilation-bounded minimum
spanning tree problem correctly answers the \textsc{Knapsack}
instance.

By construction, the points $r, a_{i}$, $b_{i}$, $d_j$ have integer
coordinates.  We construct the points $\tilde{c}_i$ by solving a
quadratic equation with an error of at most $2^{-k}$, that is, with
$k$~bits after the binary point, where $k$ is chosen such that $2^{k}
> 600nL$.  Clearly $k$ is polynomial in the input size.  If we
multiply all point coordinates in our construction and the cost bound
by~$2^{k}$, then all points have integer coordinates.

\subsubsection*{Acknowledgments.}

We are grateful to Joachim Gudmundsson for helpful discussions, in
particular about the approximation algorithm.

\bibliographystyle{plain}
\bibliography{adbmst}
\end{document}